\documentclass[letterpaper,11pt]{article}

\usepackage{latexsym}
\usepackage{verbatim}
\usepackage{url}

\usepackage{amsmath}
\usepackage{amssymb}
\usepackage{amsfonts}
\usepackage{graphicx}
\usepackage{xcolor}

\newtheorem{theorem}{Theorem}[section]
\newtheorem{lemma}[theorem]{Lemma}

\newtheorem{corollary}[theorem]{Corollary}
\newtheorem{definition}[theorem]{Definition}
\newtheorem{conjecture}[theorem]{Conjecture}

\newenvironment{proof}{{\bf Proof:\ }}{\hfill$\Box$\medskip}

\newcommand{\Exp}{\mathbb{E}}
\renewcommand{\Pr}{\mathbb{P}}

\newcommand{\ee}{{\rm e}}

\newcommand{\ignore}[1]{}
\newcommand{\remove}[1]{}
\newcommand{\reals}{\hbox{$\rlap{\rm I} \> \kern-.2mm{\rm R}$}}

\newcommand{\eij}{\{i,j\}}
\newcommand{\xij}{x(\{i,j\})}

\newcommand{\C}{{\mathcal C}}

\newcommand\blue[1]{#1}

\begin{document}

\title{Random $k$-out subgraph leaves only $O(n/k)$ inter-component edges
\thanks{In connection with this work, all authors had some degree of support
from Thorup's Investigator Grant 16582, Basic Algorithms Research
Copenhagen (BARC), from the VILLUM Foundation. All were at BARC when
the work was started, Valerie King as a short term visitor, Uri Zwick
and Or Zamir as long term visitors, and Jacob Holm and Mikkel Thorup
were there the whole time.}}





\author{
Jacob Holm\thanks{
BARC, Department of Computer Science, University of Copenhagen, Denmark, {\tt jaho@di.ku.dk, mikkel2thorup@gmail.com}}
 \and Valerie King\thanks{Department of Computer Science,
University of Victoria, Canada. E-mail: {\tt val@uvic.ca}. Research supported by Natural Science and Engineering Research Council of Canada (NSERC) Discovery Grant.} \and Mikkel Thorup${}^*$ \and Or Zamir\thanks{Blavatnik School of
Computer Science, Tel Aviv University,
  Israel. E-mail: {\tt orzamir@mail.tau.ac.il,
    zwick@tau.ac.il}. Research supported by a grant from The Blavatnik Computer Science Research Fund.} \and
Uri Zwick${}^\ddagger$}

\date{}

\maketitle

\begin{abstract}\noindent Each vertex of an arbitrary simple graph on~$n$ vertices chooses $k$ random incident edges. What is the expected number of edges in the original graph that connect different connected components of the sampled subgraph? We prove that the answer is $O(n/k)$, when $k\ge c\log n$, for some large enough~$c$. We conjecture that the same holds for smaller values of~$k$, possibly for any $k\ge 2$. Such a result is best possible for any $k\ge 2$. As an application, we use this sampling result to obtain a one-way communication protocol with \emph{private} randomness for finding a spanning forest of a graph in which each vertex sends only~${O}(\sqrt{n}\log n)$ bits to a referee.
\end{abstract}



\section{Introduction}\label{S-intro}

Sampling edges is a natural way of trying to infer properties of a graph when accessing the whole graph is either not possible or too expensive.
We consider a scenario in which each vertex is resource constrained and can only sample $k$ of its incident edges, or all edges if its degree is at most~$k$. This corresponds to the \emph{$k$-out} model that was mainly studied in the context of the \emph{complete} graph. (See references and discussion in Section~\ref{sub-kout-prev}.) Here we are interested in properties of this sampling model when applied to \emph{arbitrary} simple graphs.

Let $G=(V,E)$ be an arbitrary simple graph on $n$ vertices. Each vertex $v$ independently picks $\min\{\deg(v),k\}$ random adjacent edges. Let $G'=(V,E')$ be the resulting subgraph. How many edges of~$G$ connect different connected components of~$G'$? (These edges are referred to as \emph{inter-component} edges.) We prove that for $k\ge c\log n$, for a sufficiently large constant~$c$, the expected number of such edges is $O(n/k)$. We conjecture that the same result also holds for much smaller values of~$k$, possibly even for every $k\ge 2$. The statement is false for $k=1$. No such result was obtained or conjectured before, for any value of~$k$. Simple examples show that this result is best possible for any $k\ge 2$.
The proof we provide is fairly intricate. Our result also sheds light on other sampling models.


Given its generality, we hope that our new sampling theorem would find many applications. As a first such application, we show how the sampling theorem, together with other ideas, can be used to obtain a one-way communication protocol with \emph{private} randomness for finding a spanning forest of an input graph in which each vertex sends only~$\tilde{O}(\sqrt{n})$ bits to a referee. No private randomness protocol in which each vertex sends only $o(n)$ bits was known before.


\subsection{Our results}\label{sub-k}

We begin with a formal definition of the $k$-out model.

\begin{definition}
[Random $k$-out subgraphs]
\label{D-k} Let $G=(V,E)$ be a simple undirected graph. Sample a subset $S\subseteq E$ of the edges by the following process: Each vertex independently chooses $\min\{k,\deg(v)\}$ of its incident edges, each subset of this size is equally likely. An edge is included in~$S$ if and only if it was chosen by at least one of its endpoints. The subgraph $G'=(V,S)$ is said to be a random $k$-out subgraph of~$G$.
\end{definition}

In the above definition, we treat each undirected edge $\{u,v\}\in E$ as two directed edges $(u,v)$ and $(v,u)$. Each one of these directed edges is sampled independently. At the end, the direction of the sampled edges is ignored and duplicate edges are removed.

Although each vertex chooses only $k$ adjacent edges, the resulting subgraph is \emph{not} necessarily of maximum degree~$k$, as an edge may be chosen by either of its endpoints. In particular, if $G$ is a star and $k\ge 1$, then~$G'$ is always the original graph, as each leaf must choose the edge connecting it to the center. The choices made by the center are irrelevant. However, for any graph $G$, the resulting subgraph~$G'$ is always \emph{$k$-degenerate}, i.e., it can be oriented so that the outdegree of each vertex is at most~$k$. We just keep the orientation of the sampled directed edges.) As a consequence the \emph{arboricity} of $G'$ is also at most~$k$.

The main result of this paper is:

\begin{theorem}[Main Theorem for $k$-out]\label{T-main}
  Let $G$ be an arbitrary undirected $n$-vertex graph and let $k\ge c\log n$, where~$c$ is a large enough constant. Let $G'$ be a random $k$-out subgraph of~$G$. 
  Then the expected number of edges in~$G$ that connect different connected components of~$G'$ is $O(n/k)$.
\end{theorem}

It is easy to see that the theorem is best possible for any $k\ge 2$ and $n\ge 3k$. Let $G=(V,E)$ be a graph composed of two cliques of size $\frac{n}{2}$, connected by a matching of size $\frac{n}{k}$.
With probability at least $(1-\frac{2k}{n})^{\frac{2n}{k}} \ge (1-\frac{2}{3})^{2\cdot 3} = 3^{-6}$, no edge from the matching is chosen, in which case all the $\frac{n}{k}$ edges of the matching are inter-component, i.e., connect different connected components of~$G'$.

Another example, of a very different nature, that shows that Theorem~\ref{T-main} is best possible is the following. Let $T$ be an arbitrary tree on $\frac{n}{2k}\ge 2$ vertices. Form $G$ by connecting each vertex $u$ of~$T$ to $2k-1$ new leaves, making $\deg(u)\geq 2k$. Each original tree edge $(u,v)$ has probability at least $1-\frac{k}{\deg(u)}\ge\frac{1}{2}$ of not being chosen by $u$, and (independent) probability at least $1-\frac{k}{\deg(v)}\ge\frac{1}{2}$ of not being chosen by $v$. Thus the probability of $(u,v)$ not being in $G'$ is $(1-\frac{k}{\deg(u)})\cdot(1-\frac{k}{\deg(v)})\ge\frac{1}{4}$.
The expected number of edges of~$T$ that connect different connected components of $G'$ is therefore at least $(\frac{n}{2k}-1)\cdot\frac{1}{4}$, which is $\Omega(\frac{n}{k})$.

It follows immediately from Theorem~\ref{T-main} that there is a constant $b$ such that the probability that the number of inter-component edges is greater than $\ell\cdot bn/k$ is at most $2^{-\ell}$, for every $\ell\ge 1$, and this tail bound is tight. (See Corollary~\ref{C-tail}.)

We conjecture that Theorem~\ref{T-main} holds whenever $k=\Omega(1)$, and possibly even for every $k\ge 2$.

\begin{conjecture}[Conjecture for $k$-out]
  Let $G$ be an arbitrary undirected $n$-vertex graph and let $k\ge c$, where~$c$ is a large enough constant. Let $G'$ be a random $k$-out subgraph of~$G$. 
  Then, the expected number of edges in~$G$ that connect different connected components of~$G'$ is $O(n/k)$.
\end{conjecture}


A closely related sampling model, in which we do most of the work, is the following:

\begin{definition}[Random \emph{expected} $k$-out subgraphs]\label{D-kdegv}
Let $G=(V,E)$ be a simple undirected graph. Sample a subset
$S\subseteq E$ of the edges by the following process: Each vertex
samples each one of its incident edges independently with probability
$k/\max\{k,\deg(v)\}$. Thus, each vertex of degree at least~$k$ samples an expected number of~$k$ edges.
An edge is included in~$S$ if and only if it
was sampled by at least one of its endpoints. The subgraph $G'=(V,S)$
is said to be a random \emph{expected} $k$-out subgraph of~$G$.
\end{definition}


Let $e=\{u,v\}\in E$ be an edge. If either $u$ or~$v$ are of degree at most~$k$, then $e$ is always sampled. Otherwise, $e$ is sampled with probability $p_e=\frac{k}{\deg(u)}+\frac{k}{\deg(v)} - \frac{k}{\deg(u)}\cdot\frac{k}{\deg(v)}$.
Equivalently, if we view~$e$ as two directed edges $e'=(u,v)$ and $e''=(v,u)$, then $e'$ is sampled with probability $p_{e'}=\frac{k}{\deg(u)}$ and $e''$ is sampled with probability $p_{e''}=\frac{k}{\deg(v)}$. (Recall that the directions of the sampled directed edges are ignored.)
Choices made for different edges are completely independent.




We show below (Lemma~\ref{lem:same-expectation}) that when $k\ge c\log n$, for a sufficiently large constant $c$, the expected number of inter-component edges with respect to an expected $k$-out subgraph is essentially sandwiched between the corresponding expectations for (exact) $k/2$-out and (exact) $2k$-out subgraphs.

To prove Theorem~\ref{T-main} for $k\ge c\log n$, for a sufficiently large~$c$, it is thus sufficient to prove the following theorem, which we do in Section~\ref{S-proof}.

\begin{theorem}[Main Theorem for \emph{expected} $k$-out]\label{T-variant}
  Let $G$ be an arbitrary undirected $n$-vertex graph and let $k\ge c\log n$, where~$c$ is a large enough constant. Let $G'$ be a random expected $k$-out subgraph of~$G$. 
  Then, the expected number of edges in~$G$ that connect different connected components of~$G'$ is $O(n/k)$.
\end{theorem}

The requirement $k\ge c\log n$ in Theorem~\ref{T-variant} is essential, and thus the bound in the theorem is best possible for the expected $k$-out model. Let $G=K_n$ be the complete graph on~$n$ vertices and let $k=c\ln n$. The expected $k$-out model is then equivalent to the classical $G(n,p)$ model with $p\sim \frac{2c\ln n}{n}$. The probability that a given vertex is isolated is then $(1-\frac{2c\ln n}{n})^{n-1}\approx n^{-2c}$ and the expected number of edges connecting different connected components is $\Omega(n^{2(1-c)})$. Thus, the claim of the theorem is false when $c<\frac{1}{2}$.

This also explains the difficulty of extending Theorem~\ref{T-main} to
the $k=o(\log n)$ regime.
It is conceivable that Theorem~\ref{T-main} holds for any $k\ge 2$.

As an interesting corollary of~Theorems~\ref{T-main} and~\ref{T-variant} we get:

\begin{corollary}[Random $k$-out subgraph of a $(cn/k)$-edge connected graph]\label{C-con}\sloppy Let $G$ be a $(cn/k)$-edge connected $n$-vertex graph and let $k\ge c\log n$, where $c$ is a large enough constant. Let $G'$ be a $k$-out or an expected $k$-out random subgraph of~$G$. Then $G'$ is \emph{connected} with probability at least $1/2$.
\end{corollary}

\begin{proof} Let $p$ the probability that $G'$ is not connected. As $G$ is $(cn/k)$-edge connected, if $G'$ is not connected, then the number of inter-component edges is at least $cn/k$, and the expected number of inter-component edges is at least $p\cdot cn/k$. By Theorem~\ref{T-main} or~\ref{T-variant} this expectation is at most $bn/k$, for some constant $b$.
Thus, $p\le b/c$ and the result follows if~$c\ge 2b$. We also need~$c$ to be large enough for Theorem~\ref{T-main} or~\ref{T-variant} to hold.
\end{proof}

Both the $k$-out and expected $k$-out models favor the selection of edges incident to low degree vertices. In Section~\ref{sub-class} we compare the $k$-out sampling model to the standard model of picking each edge with some fixed probability~$p$ and explain why the $k$-out model gives much better results, in certain cases, using the same total number of sampled edges.

On regular or almost regular graphs, the $k$-out sampling model is essentially identical to the model of sampling each edge independently with probability $p=kn/m$. Surprisingly, Theorem~\ref{T-variant} implies a new result in this model for almost regular graphs, see Theorem~\ref{T-regular} below.

An appealing feature of the $k$-out model is that the sampling can be implemented in a distributed manner, as the choices of different vertices are independent. Requiring each vertex to choose only~$k$ edges is a natural constraint in many settings, e.g., if a vertex has to \emph{communicate} the edges it selected to other vertices or to a referee. It is exactly such a setting (see Section~\ref{sub-app}) that motivated us. 
However, we believe that the new sampling theorems have importance beyond the concrete applications we give here.

In many settings, including the application described in Section~\ref{sub-app}, the number of inter-component edges is a measure of the ``work'' that still needs to be done after ``processing'' the sampled subgraph.

\subsection{Previous results in the $k$-out model}\label{sub-kout-prev}

The $k$-out model is first mentioned in a question of Ulam in ``The Scottish Book''\footnote{``The Scottish Book'' was a  notebook used  in the 1930's and 1940's by mathematicians of the Lw\'{o}w School of Mathematics in Poland to collect problems. The notebook was named after the ``Scottish Caf\'{e}" where it was kept. Among the contributors to the book were  Stefan Banach, 
John von Neumann 
and Stanislaw Ulam.}~\cite{mauldin1981scottish}.

\begin{quote}
PROBLEM 38: ULAM \\
Let there be given $N$ elements (persons). To each element we attach $k$ others
among the given $N$ at random (these are friends of a given person). What is the
probability $P_{k,N}$ that from every element one can get to every other element through
a chain of mutual friends? (The relation of friendship is not necessarily symmetric!)
Find $\lim_{N\to\infty} P_{k,N}$ (0 or 1?).
\end{quote}

While this explicitly defines a \emph{directed} model, most answers in the literature are for the corresponding undirected model. It is not difficult to prove, see \cite{mauldin1981scottish}, that for $k\ge 2$ the resulting undirected graph is connected with probability tending to~$1$, while for $k=1$ the graph is connected with probability tending to~$0$.

Let $G_{k\text{-out}}$ be a random $k$-out subgraph of the complete
graph on~$n$ vertices, as in Definition~\ref{D-k}. Fenner and Frieze
\cite{FennerF82} prove that for $k\ge 2$, $G_{k\text{-out}}$ is
$k$-vertex and $k$-edge connected with probability tending to~$1$, as
$n$ tends to~$\infty$. Frieze \cite{frieze1986maximum} proved that
when $n$ is even then $G_{k\text{-out}}$ has a perfect matching with
probability tending to~$1$, if $k\ge 2$, and tending to~$0$ if
$k=1$. Bohman and Frieze~\cite{bohman2009hamilton} prove that
$G_{k\text{-out}}$ has a Hamiltonian cycle with probability tending
to~$1$, if $k\ge 3$, and tending to~$0$, if $k=1,2$. All these results
can also be found in a chapter on random $k$-out graphs in the book of
Frieze and Karo{\'n}ski \cite{frieze2016introduction}.



Frieze et al.~\cite{FriezeGRV14} consider a random subgraph obtained by taking an arbitrary spanning forest of a graph~$G$, and $k-1$ random outgoing edges from each vertex. They prove that the resulting random subgraph has some desirable expansion properties with probability tending to~$1$.

Frieze and Johansson \cite{FriezeJ17} consider random $k$-out
subgraphs of graphs of minimum degree $(\frac{1}{2}+\varepsilon)n$,
for some $\varepsilon>0$. They show that if $2\le k = o(\sqrt{\log
  n})$, then the random $k$-out subgraph is $k$-connected with
probability tending to~$1$. Thus they generalize the earlier
results of Fenner and Frieze
\cite{FennerF82} for a complete base graph to arbitrary base graphs with
sufficiently high minimum degree. Frieze and Johansson \cite{FriezeJ17}
points out that the generalization fails for lower degrees: there
are connected graphs with minimum degree $n/2$ where a random $k$-out
subgraph is not even expected to be connected.

Our results are quite different from all the results cited above. We consider random $k$-out subgraphs of an \emph{arbitrary} base graph~$G$. As the graph~$G$ is arbitrary, we cannot expect the random $k$-out subgraph to be connected, with high probability. We focus instead on the question of how closely a random $k$-out subgraph of~$G$ captures the connectivity of~$G$. We do that by bounding the expected number of \emph{inter-component} edges, i.e., the number of edges of~$G$ that connect different connected components of the sampled subgraph.

To the best of our knowledge, no result similar to our Corollary~\ref{C-con} was known before. It replaces the requirement of a very high minimum degree made in Frieze and Johansson \cite{FriezeJ17} by a much weaker connectivity requirement. However, the resulting random $k$-out subgraph is only guaranteed to be connected with probability $1/2$, not with a probability tending to~$1$. This is best possible.


In a very recent paper \cite{GNT19}, Ghaffari et al.~used 2-out sampling
to get faster randomized algorithms for edge connectivity. One of their lemmas
is that, with high probability, the number of components
in a random 2-out subgraph is $O(n/\delta)$ where $\delta$ is the smallest
degree. The same bound on the number of components is tight for $k$-out
for any $k\geq 2$. This result complements our bound on the
number of inter-component edges, and may inspire further investigations
into the properties of random $k$-out subgraphs.

\subsection{Sampling each edge independently with probability~$p$}\label{sub-class}

The most widely studied random graph model is, of course, $G(n,p)$, in which each edge of the complete graph on~$n$ vertices is sampled, independently, with probability~$p$. There are literally thousands of papers written on such random graphs.

The $G(n,p)$ model can also be used to construct a random subgraph of a general base graph $G=(V,E)$. The most relevant result to our study is the following theorem:

\begin{theorem}[Karger, Klein and Tarjan \cite{KargerKT95}]\label{T-KKT}
  Let $G=(V,E)$ be an arbitrary simple graph and let $0<p<1$. Let $G'=(V,E')$ be a random subgraph of~$G$ obtained by selecting each edge of~$G$ independently with probability~$p$. Then, the expected number of edges of~$G$ that connect different connected components of~$G'$ is at most $n/p$.
\end{theorem}

Theorem~\ref{T-KKT} is a special case of a theorem of Karger et al.~\cite{KargerKT95} that deals with weighted graphs. The more general theorem states that if $F$ is a \emph{minimum spanning forest} of~$G'$, then the expected number of edges in~$G$ that are \emph{$F$-light}, i.e., can be used to improve~$F$, is at most $n/p$. When the graph is unweighted, i.e., all edge weights are~$1$, an edge is $F$-light if and only if it connects different connected components of~$G'$. An alternative proof, of the weighted version, using backward analysis was obtained by Chan \cite{Chan98}. The weighted theorem was used by Karger et al.~\cite{KargerKT95} to obtain a linear expected time algorithm for finding Minimum Spanning Trees.

Theorem~\ref{T-KKT}, as stated, was used by Karger et al.~\cite{KargerNP99} and Halperin and Zwick~\cite{HalperinZ01} to obtain almost optimal and then optimal randomized EREW PRAM algorithms for finding connected components and spanning trees.

\subsubsection{Comparing $k$-out sampling and independent $p$ sampling}
Let us compare our new Theorem~\ref{T-main} with Theorem~\ref{T-KKT}. Let $G=(V,E)$ be a general $n$-vertex $m$-edge graph. Theorem~\ref{T-KKT} produces a sample of expected size $pm$. To obtain a random subgraph with the same number of edges we choose $k=pm/n$. The expected number of edges connecting different connected components is then $O(n/k) = O((n^2/m)/p)$. Note that this is a huge improvement over the $n/p$ bound of Theorem~\ref{T-KKT} when the graph is dense, i.e., $m\gg n$. Alternatively, if we express the expressions in term of~$k$, the bound for $k$-out is $O(n/k)$, while the bound for independent sampling is only $n/p=m/k$.

To highlight the difference between the two sampling schemes, and to show that the gap between $O((n^2/m)/p)$ and $n/p$ can actually occur, consider the following situation. Let $G$ be a graph composed of a clique of size $n/2$ and $8k$ cliques of size $n/(16k)$. All cliques are disjoint. The number of edges is $m\ge \frac{1}{8}n^2$. If $pm=kn$, then $p\le \frac{8k}{n}$. Consider a vertex in one of the small cliques. With a constant probability none of its incident edges are sampled, in which case it contributes $\frac{n}{16k}-1$ to the expectation. The expected number of edges connecting different connected components is $\Omega(n^2/k)$. Theorem~\ref{T-KKT} is thus asymptotically tight in this case. The corresponding bound for $k$-out is $O(n/k)$, a factor of~$n$ smaller. This example shows that it is much wiser, in certain situations, to sample edges incident on low degree vertices with higher probabilities.




\subsubsection{Improved result for independent $p$ sampling for almost regular graphs}

While Theorem~\ref{T-KKT} is best possible for general graphs, we show that it can be improved for \emph{almost regular} graphs.

\begin{definition}[Almost regular graphs]
  A graph $G=(V,E)$ is said to be almost $r$-regular if $r\le \deg(v)\le 2r$, for every $v\in V$.
  A graph is almost regular if it is almost $r$-regular for some~$r$.
\end{definition}

If $G$ is $r$-regular, then expected $k$-out sampling is equivalent to sampling each edge independently with probability $p=k/r$. The models are very closely related if~$G$ is almost $r$-regular. Namely, $k$-out sampling produces a subsample of the sample obtained by sampling each edge with probability $p=k/r$. Thus, Theorem~\ref{T-variant} immediately implies the following new result for independent $p$ sampling.

\begin{theorem}[Independent $p$ sampling of almost regular graphs]\label{T-regular}
  Let $G=(V,E)$ be an almost regular graph and let $0<p<1$. Let $G'=(V,E')$ be a random subgraph of~$G$ obtained by selecting each edge of~$G$ independently with probability~$p$. Then, the expected number of edges of~$G$ that connect different connected components of~$G'$ is $O((n^2/m)/p)$.
\end{theorem}

We note that Theorem~\ref{T-regular} cannot be extended to the weighted case. Consider a complete weighted graph on~$n$ vertices in the which the weights of all edges incident on a vertex are distinct. It is easy to see that the expect number of $F$-light edges is $\Theta(n/p)$.

\subsection{Applications in distributed computing}\label{sub-app}

The problem that led us to consider the sampling model discussed in this paper is the following. Each vertex in an undirected graph only knows its neighbors.
It can send a single message to a referee which should then determine, with high probability, a spanning forest of the graph. How many bits does each vertex need to send the referee?

If the vertices have access to public randomness, the answer is $\Theta(\log^3 n)$. The upper bound follows easily from Ahn et al.  \cite{ahn2012graph}. A matching lower bound was recently obtained by Nelson and Yu \cite{NY2018}.

In Section~\ref{S-sfprivate} we show, using the sampling theorem (Theorem~\ref{T-main}) with $k=\sqrt{n}$, and a few other ideas, that $O(\sqrt{n}\log n)$ bits are sufficient when the vertices only have access to \emph{private randomness}. Nothing better than $O(n)$ was previously known with private randomness.

The best known deterministic protocol uses $n/2$ bits and it is open whether this is optimal.

In the MapReduce-like model (see \cite{Karloff2010}), messages are passed in the form of $(key,value)$ pairs of $O(\log n) $ bits which is the wordsize. We assume there are $n$ machines and each can send, receive, and compute with no more than $m$ words in any round. In each round if the words with the same key can fit on one machine, then one machine will receive them all, process them and emit  a new set of messages for the next round.  Here we assume each edge $(u,v)$ of the input graph appears twice, as $(key=u, value=v)$ and $(key=v, value=u)$. As the machines are stateless, these will be recirculated in every round.   It is easily seen that if  $m=O(n \log^2 n)$, the one-way communication algorithm with public randomness cited above leads to a one round Monte Carlo algorithm in the MapReduce-like model. Each machine which receives the edges incident to a particular vertex will compute the corresponding $O(\log^3 n)$-bit message for the referee and send it using $O(\log^2 n)$ $O(\log n)$-bit messages all tagged with a special key. One machine will receive all these messages from all machines and can then act as the referee to compute a spanning forest. If no public randomness is available, then this must be preceded by a round in which a random string is created by one machine and a copy is sent with a key for each vertex.

A consequence of Theorem \ref{T-main} is an almost equally simple four-round algorithm to compute a spanning forest which requires space of only $m=O(n k )$ words. Note that for $k=\log n$, this is a factor $\log n$ less memory than above. As before, for each vertex, there is a machine which receives the edges incident to the vertex and it sends out, in round~1, its up to $k$ sampled edges to a referee. This is done by giving all these messages a special key so that they are received by one machine that can act as the referee. The referee computes the spanning forest of the $k$-out subgraph. The spanning forest can be distributed to $n$ machines in two rounds (see {Jurdzi\'{n}ski} and Nowicki \cite{Jurdzinski2018}). We assume again that the input graph is recirculated and for each vertex  there is a machine which receives all its incident edges and the spanning forest of the $k$-out subgraph. This machine can determine which of its incident edges connect up different components of the $k$-out subgraph, and use a special key to send them out. The spanning forest of the $k$-out subgraph is also recirculated with the special key.  Since there are no more than $O(n/k) +n <m $ such edges, a single machine will receive all these inter-component edges and the edges of the $k$-out spanning forest and compute the spanning forest of the original graph. If Theorem~\ref{T-main} holds with $k=O(1)$, this would be an extremely simple four-round algorithm with $m=O(n)$.

Jurdzi\'{n}ski and Nowicki \cite{Jurdzinski2018} obtained a $O(1)$ round algorithm with $m=O(n)$, but it is much more complicated.

\section{Proof of the sampling theorem}\label{S-proof}

\subsection{Relation between $k$-out and expected $k$-out models}



%
%
%


The following simple lemma shows that for sufficiently large $k$, the expected and (exact) $k$-out models have essentially the same expected number of inter-component edges.

\begin{lemma}\label{lem:same-expectation}
  Let $G$ be an arbitrary undirected $n$-vertex graph. For any $t$,
  let $X_t$ be the number of inter-component edges in $G$ with respect to a $t$-out subgraph, and let $Y_t$ be the number of inter-component edges in $G$ with respect to an expected $t$-out subgraph.
  Then for
  $k\geq c\log n$ where $c$ is a large enough constant,
  $\Exp[X_{2k}]-o(1) \leq \Exp[Y_k] \leq
  \Exp[X_{\frac{k}{2}}]+o(1)$.
\end{lemma}
\begin{proof}
Consider the directed edges in a random expected $k$-out subgraph.
For any  $v\in V$, let~$S_{v,k}$ be the set of outgoing edges of~$v$ in this subgraph, and let~$s_{v,k}=|S_{v,k}|$. Note that given $s_{k,v}$,
$S_{k,v}$ is a random subset of the outgoing edges of~$v$ of size $s_{v,k}$.
If $\deg_G(v)\le k$, then $S_{k,v}$ includes all edges incident to~$v$. Assume that $\deg_G(v)> k$.
Let $E_v$ be the set of outgoing edges of~$v$ in~$G$.
If $s_{v,k}\ge k/2$, then a random subset of size $k/2$ of $S_{k,v}$
is a random subset of~$E_v$ of size $k/2$. Thus, it has the same distribution as the subset of edges chosen by~$v$ in the exact $k/2$-out model. Similarly, if $s_{v,k}\le 2k$, then a random subset of~$E_v$ of size $2k$ that contains $S_{k,v}$ has exactly the same distribution as the edges chosen by~$v$ in the exact $2k$-out model.

Let
  $\mathcal{E}$ be the event that $\frac{k}{2} \leq S_{v,k} \leq 2k$
  for every vertex $v$ with $\deg_G(v)>k$.  Then
  \begin{align}
    \Exp[X_{2k}]
    \:\leq\:
    \Exp[Y_k \mid \mathcal{E}]
    \:\leq\:
    \Exp[X_{\frac{k}{2}}]\;.
    \label{eq:expectation-bounds}
  \end{align}
  The first inequality holds because assuming $\mathcal{E}$, we can
  add random edges to the expected $k$-out subgraph to get an
  exact $2k$-out subgraph, as explained above. Adding random edges to a sampled subgraph can only decrease the number of inter-component edges. Similarly, the second inequality holds because
  assuming $\mathcal{E}$, we can randomly remove edges from the expected
  $k$-out subgraph to get an exact $\frac{k}{2}$-out subgraph, and this
  can only increase the number of inter-component edges.
  It is easy to check that the required containments hold also for vertices of degree less than~$k$.

  Let $p=\Pr[\neg\mathcal{E}]=1-\Pr[\mathcal{E}]$ be the probability
  that $S_{v,k}\not\in[\frac{k}{2},2k]$ for at least one $v$ with
  $\deg_G(v)>k$. Then $\Exp[Y_k] = (1-p)\Exp[Y_k\mid\mathcal{E}] +
  p\Exp[Y_k\mid\neg\mathcal{E}]$ and thus $\Exp[Y_k\mid \mathcal{E}]
  = \Exp[Y_k] +
  \frac{p}{1-p}\big(\Exp[Y_k]-\Exp[Y_k\mid\neg\mathcal{E}]\big)$.  Setting
  $D := \frac{p}{1-p}\left|\vphantom{x^{x^x}_{y}}\Exp[Y_k]-\Exp[Y_k\mid\neg\mathcal{E}]\right|$ and combining with~(\ref{eq:expectation-bounds}) gives
  \begin{align*}
    \Exp[X_{2k}]-D
    \leq
    \Exp[Y_k\mid\mathcal{E}] - D
    \leq
    \Exp[Y_k]
    \leq
    \Exp[Y_k\mid\mathcal{E}] + D
    \leq
    \Exp[X_{\frac{k}{2}}] + D \;.
  \end{align*}
  By definition, $0\leq Y_k\leq \binom{n}{2}$, so $D\leq
  \frac{p}{1-p}\binom{n}{2}$.  For any vertex $v$ with
  $\deg_G(v)>k$, we have $\Exp[S_{v,k}]=k$, so by Chernoff bounds
  $\Pr[S_{v,k}>2k]<e^{-\frac{k}{3}}$ and
  $\Pr[S_{v,k}<\frac{k}{2}]<e^{-\frac{k}{8}}$. A union bound then
  gives $p< n\cdot(e^{-\frac{k}{3}}+e^{-\frac{k}{8}})$. For $k\geq c\ln n$, where $c>24$, we then have
  $D\leq\frac{p}{1-p}\binom{n}{2}=o(1)$.
\end{proof}

The first inequality $\Exp[X_{2k}]\le \Exp[Y_k]+o(1)$, for $k\ge c\log n$, shows that Theorem~\ref{T-variant} implies Theorem~\ref{T-main}.

\subsection{Overview of the proof of the sampling theorem for expected $k$-out}\label{sub-overview}



\newcommand{\EE}{\mathbb{E}}

We prove Theorem~\ref{T-variant} by sampling edges gradually.
We construct two sequences of random subgraphs $G_1\subseteq G_2\subseteq\ldots\subseteq G_r$, and $H_1\supseteq H_2\supseteq\ldots\supseteq H_r$, for some $r\le n$, all subgraphs of~$G$. We also have $G_i\subseteq H_i$, for $i=1,2,\ldots,r$.  Each subgraph $G_i$ is obtained by sampling more edges from $H_{i-1}\setminus G_{i-1}$ and \emph{adding} them to~$G_{i-1}$. Each subgraph $H_i$ is obtained from~$H_{i-1}$ by \emph{removing} some edges that are inter-component edges of~$G_i$, and possibly some additional edges not contained in~$G_i$. Initially $G_1$ contains all edges that touch a vertex of degree at most~$k$. Note that all these edges are contained in every (expected) $k$-out sample.
We let $H_1=G$, the original graph. At the end of the process, $G_r$ and $H_r$ have the same connected components, viewed as sets of vertices.

Edges are sampled and added during the process in a way that ensures that $G_r$ can be extended into an expected $k$-out sample of~$G$. Adding more edges to~$G_r$ can only reduce the number of inter-component edges. All the edges of~$G$ that are inter-component with respect to~$G_r$, and therefore also with respect to~$H_r$, must be in $G\setminus H_r$. Thus, to prove the theorem it is enough to prove that $\EE[|G\setminus H_r|]=O(n/k)$.

\paragraph{The round sampling model}

%
A convenient way to view the expected $k$-out sampling is that for every (directed) edge $e=(u,v)\in E$, where $\deg(u)\ge k$, we independently
generate a uniform random number $x_e\in [0,1]$, and include (the undirected version of)~$e$ in the sampled subgraph if and only if $x_e< p_e=k/\deg(u)$. (In what follows, we let $E$ stand for both the undirected edges of the original graph $G$ and their corresponding directed edges.)

Edges are sampled in rounds.
For every edge $e$, we generate a sequence of numbers
$q_{e,0},\ldots, q_{e,r}$ such that
$0=q_{e,0}\leq q_{e,1}\leq \cdots \leq q_{e,r}=p_e$, 
where $r$ is the number of rounds.
Edge~$e$
is sampled in round~$i$ if and only if $x_e\in [q_{e,i-1},q_{e,i})$.
An edge $e$ is thus sampled in any one of the rounds, if and only if $x_e<p_e$.
The last round $r$ is special and is just used to ensure that $p_{e,r}=p_e$ for every $e$.

In the beginning of round~$i$, where $i< r$, we set $q_{e,i}$ for every (directed) edge $e\in E$ such that $q_{e,i}\in [q_{e,i-1},p_e)$. The values $q_{e,i}$ may depend on the sampling outcomes in previous rounds.
That is, for all $j<i$, we
know $G_j$, $H_j$, and all the $q_{e,j}$. In particular, we know if~$e$ has already been sampled, i.e., if $x_e<q_{e,i-1}$.

We do not choose the number of rounds $r$ in advance.
At the end of round $i$, we may decide to end round sampling with one last special round, by setting $r=i+1$
and $q_{e,r}=p_e$ for all edges $e$.

Given that
$e$ was not sampled 
in rounds $1,\ldots,i-1$, the probability that \blue{it is} sampled 
in round~$i$
is $p_{e,i}=(q_{e,i}-q_{e,i-1})/(1-q_{e,i-1})$. In the analysis, we adopt the opposite view: We choose
$p_{e,i}$ and then set $q_{e,i}=q_{e,i-1}+p_{e,i}(1-q_{e,i-1})$. We can pick any sequence of $p_{e,i}$
such that $\sum_{j=1}^i p_{e,j}\leq p_e$.

\paragraph{General sampling strategy}

More specifically, to move from~$G_i$ to~$G_{i+1}$, we look at a \emph{smallest} connected component $A_i$ of~$G_i$ that is \emph{not} a complete connected component of~$H_i$. We assign edges that emanate from~$A_i$ some carefully chosen non-zero sampling probabilities $p_{e,i}$, being careful not to exceed the overall sampling probability of each edge. We add edges to~$G_{i+1}$ according to these probabilities. We hope that the addition of new edges, if any, connects~$A_i$ to at least one other component of~$G_i$. If this happens, i.e., at least one of the edges of the cut $\partial_{H_i}(A_i)$ defined by~$A_i$ in~$H_i$ was sampled, we let $H_{i+1}=H_i$. Otherwise, we let $H_{i+1}$ be~$H_i$ with the edges of the cut $\partial_{H_i}(A_i)$ removed. (Note that this ensures that $A_i$ is also a connected component of~$H_{i+1}$.) We are essentially deciding to `give-up' on the edges of $\partial_{H_i}(A_i)$ and assume that they will end up as inter-component edges of~$G_r$.
If $H_{i+1}\subsetneq H_i$, we perform an additional \emph{trimming} operation that removes some additional edges from~$H_{i+1}$, but still maintaining $G_{i+1}\subseteq H_{i+1}$. Trimming will be explained later.

Each round reduces the number of connected components of~$G_i$ that are not connected components of~$H_i$.  Thus, for some $r\le n$, all the connected components in~$G_{r}$ are also connected components of~$H_{r}$ and the process stops.

\newcommand{\TT}{{\cal T}}

We let $X_i=|H_{i}\setminus H_{i+1}|$, the number of edges that were removed from~$H_i$ to form~$H_{i+1}$. Note that~$X_i$ is a random variable. Our goal is to bound $\EE[|G\setminus H_r|] = \EE[\sum_{i=1}^{r-1} X_i] = \sum_{i=1}^{r-1} \EE[X_i]$. To bound $\EE[X_i]$, we consider conditional expectations of the form $\EE[X_i|T_i]$, where $T_i$ is a \emph{transcript} of everything that happened during the construction of $G_1\subseteq G_2\subseteq\ldots\subseteq G_i$ and $H_1\supseteq H_2\supseteq\ldots\supseteq H_i$. (Actually, $G_1,\ldots,G_i$ and $H_1,\ldots,H_i$ give a full account of what happened. The last two subgraphs~$G_i$ and~$H_i$ on their own do not tell the whole story, as they do not specify when each edge of~$G_i$ was added.)


Let $\TT$ be the probability distribution induced on full transcripts of the process described above.  Let $\TT_i$ be the probability distribution of transcripts of the first~$i$ rounds. Given a full transcript~$T$, we let $T_i$ be its restriction to the first $i$ rounds. Clearly, if $T$ is chosen according to~$\TT$, then $T_i$ is distributed according to~$\TT_i$. Then,

\[ \sum_{i=1}^{r-1} \EE[X_i] \;=\;
\sum_{i=1}^{r-1} \EE_{T_i\sim \TT_i}\bigl[ \EE[X_i|T_i] \bigr] \;=\;
\EE_{T\sim\TT}\left[ \sum_{i=1}^{r-1} \EE[X_i|T_i] \right] \;.\]
We prove that $\sum_{i=1}^{r-1} \EE[X_i|T_i] = O(n/k)$ for \emph{every} transcript~$T$. 
This is done using judiciously chosen sub-sampling probabilities and an intricate amortized analysis that shows that charging each vertex a cost of $O(1/k)$ is enough to cover the expected total number of inter-component edges.

For a given transcript $T$, the expression $\sum_{i=1}^{r-1} \EE[X_i|T_i]$ has the following meaning. The transcript~$T$ gives a full description of what happens during the whole process. No randomness is left. However, in the~$i$th round we compute the expectation of~$X_i$ given the past $T_i$, \emph{without} peeping into the future.

(The above discussion ignores the fact that $r$ is also a random variable. This can be easily fixed by adding dummy rounds that ensure that we always have $r=n$.)

We are thus left with the task of bounding $\sum_{i=1}^{r-1} \EE[X_i|T_i]$ for a specific transcript~$T$. Let~$A_i$ be the component we try to connect at the $i$-th round. The expectation $\EE[X_i|T_i]$ is the size of cut $\partial_{H_i}(A_i)$ multiplied by the probability that no edge from this cut is sampled. (This ignores trimming that will be explained later.) We need to bound this quantity and then decide how to split it among the vertices of~$A_i$. As mentioned, the total cost of each vertex, during all iterations, should be $O(1/k)$. In deciding how to split the cost, we take the full history~$T_i$ into account.

Having explained the general picture, we need to get down to the details. In the following we use~$G'$, $H$ and $A$ instead of $G_i$, $H_i$ (or $H_{i+1}$) and $A_i$.

\subsection{Growing and trimming}\label{sub-trim}

Our proof repeatedly uses growth and trimming steps:
\begin{itemize}
\item {\bf Growth Step}: Pick a smallest connected component $A$ of $G'$ which is not a complete connected component of~$H$ and try to connect it to other components of~$G'$ by some limited sampling of edges adjacent to~$A$ (as shall be described later). If the step fails to connect $A$ to other components, remove all of the edges in the cut defined by $A$ from~$H$.
\item {\bf Trimming Step}: If $v\in G$ is such that $\deg_{H}(v) < \frac{1}{3} \deg_{G}(v)$, \emph{remove} from~$H$ all edges that connect~$v$ to vertices not in the connected component of~$v$ in~$G'$.
\end{itemize}

A schematic description of the process is given in Figure~\ref{fig:grow}. The red components on the left are components that will not grow anymore, as none of the edges going out of them is in~$H$. The green components on the right are components that may still grow. Solid edges are edges that were already sampled, i.e., are in~$G'$. Dashed edges are edges in $H\setminus G'$, i.e., edges that were not sampled yet, but may still be sampled. Finally, dotted edges are edges that were removed from~$H$ and thus will not be sampled. Unfilled vertices represent vertices that were trimmed.

\begin{figure}[t]
\begin{center}
\includegraphics[scale=0.35]{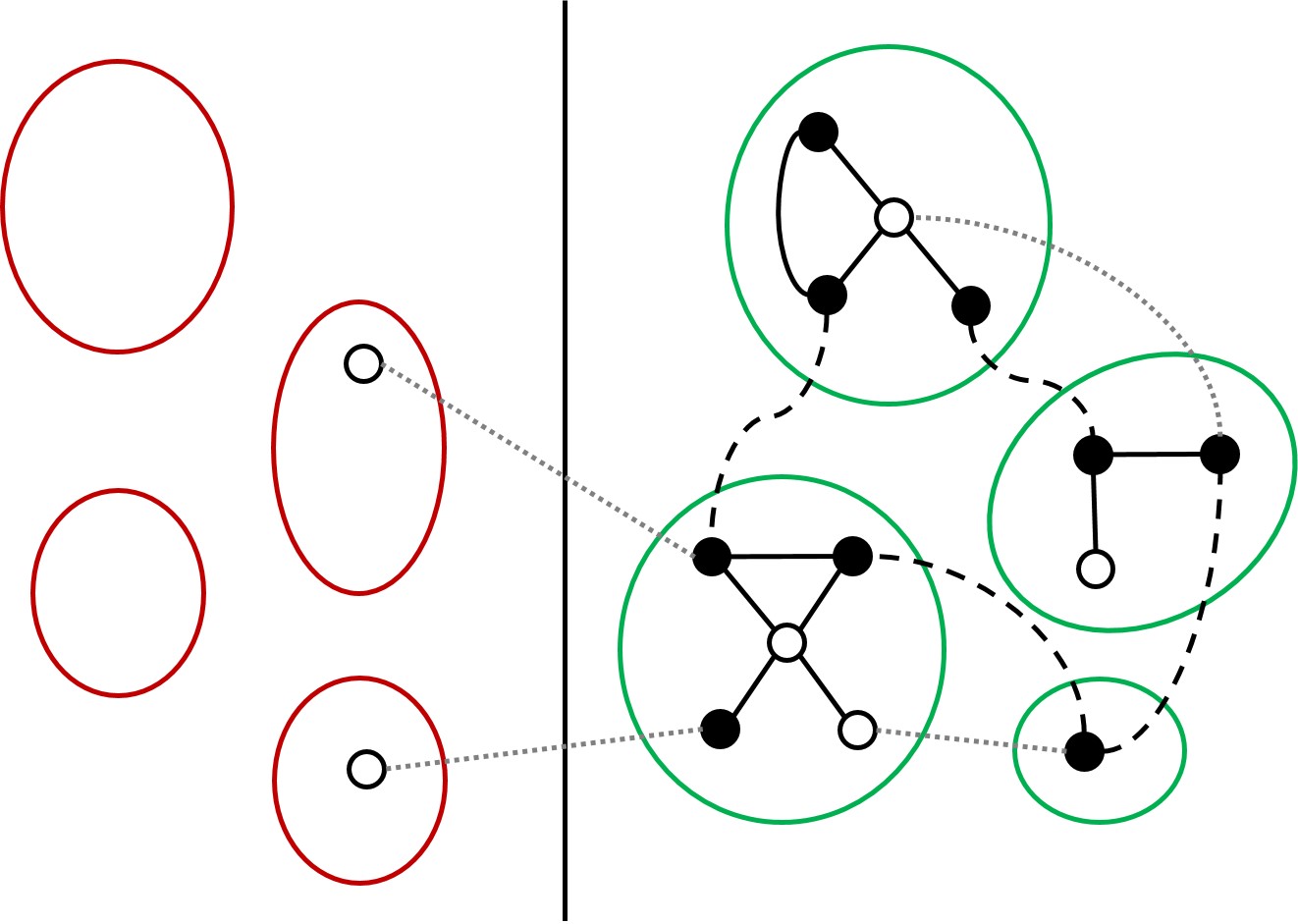}
\end{center}
\vspace*{-10pt}
\caption{A schematic description of the repeated growing and trimming process.}\label{fig:grow}
\end{figure}


By choosing a smallest component at each growth step we ensure that the size of the component is at least doubled, if the step is successful. Each vertex can therefore participate in at most $\log n$ such growth steps. (Components may also grow as a result of other components growing into them.)

A trimming step may trigger a cascade of additional trimming steps. When the process ends, for every vertex $v$ which has edges in~$H$ that connect it to vertices not in the connected component of~$v$ in~$G'$, we have $\deg_H(v)\ge \frac{1}{3}\deg_G(v)$. This is important in the proofs below as it shows that we can sample edges with probabilities proportional to $k/\deg_H(v)$. The following lemma claims that trimming does not remove too many additional edges.

\begin{lemma}
The number of edges removed from~$H$ during Trimming steps is at most the number of edges removed from~$H$ during (unsuccesful) Growing steps.
\end{lemma}
\begin{proof}
We use a simple amortization argument. We say that a vertex~$v$ is \emph{active} if $v$ is part of a component in~$G'$ that can still grow and if it was not trimmed yet. Suppose that a vertex is trimmed when $\deg_H(v)\le \alpha \deg_G(v)$ for some $\alpha<\frac{1}{2}$. (We use $\alpha=\frac{1}{3}$.) We define the \emph{potential} of an active vertex~$v$ to be $\beta(\deg_G(v)-\deg_H(v))$, i.e., $\beta$ times the number of edges of~$v$ that were deleted, for some $\beta>0$ to be chosen later. The amortized cost of deleting an edge $(u,v)$
following a failed growth step, where $u$ belongs to the component that failed to grow, is at most $1+\beta$, where~$1$ is the actual cost of deleting the edge, and $\beta$ is the increase of potential of~$v$ if it remains active. When a vertex~$u$ is trimmed, its potential is at least $(1-\alpha)\beta\deg_G(u)$. We want this potential to pay for the amortized cost of deleting the at most $\alpha \deg_G(v)$ remaining edges of~$u$, i.e., $(1-\alpha)\beta \ge \alpha(1+\beta)$, which is satisfied if we choose $\beta=\frac{\alpha}{1-2\alpha}$. For $\alpha=\frac{1}{3}$ we have $\beta=1$. It follows that the total number of edges deleted is at most twice the number of edges deleting following failed growth steps.
\end{proof}

Thus, it is enough to bound the number of edges removed during Growing steps.
Moreover, during Growing steps we can always assume that the degrees of relevant vertices in~$H$ are the same as their degrees in $G$ up to a multiplicative constant. Up to a constant factor, we can therefore assume that an edge $(u,v)$ is sampled with probability $k/\deg_H(u)$.

\subsection{High level details}\label{sub-high-level}

During growth steps, we distinguish between two types of components.
\begin{itemize}
\item {\bf High-degree:} A component~$A$ that contains a vertex $v\in A$ such that $\deg_{H}(v) \geq 2|A|$.
\item {\bf Low-degree:} A component~$A$ such that for every vertex $v\in A$ we have $\deg_{H}(v) < 2|A|$.
\end{itemize}

High-degree components are easier to deal with when $k\ge c\log n$. Note that if $\deg_H(v)\ge 2|A|$, then at least half of the edges of~$v$ leave the component~$A$. Thus, a random edge of~$v$ leaves~$A$ with a probability of least~$1/2$.

To grow a high-degree component $A$ we pick some vertex $v\in A$ with $\deg_{H}(v) \geq 2|A|$ and repeatedly sample its edges until an edge connecting it to a different connected component is found. This is a slight deviation from the framework of Section~\ref{sub-overview}, as we might exceed the sampling probabilities of certain edges. We show, however, that this only happens with a very small probability, and is therefore not problematic. (See the formal argument after Lemma~\ref{smpls1}.)

The main challenge, where novel ideas need to be used, is in the growth of low-degree components.

To grow a low-degree component~$A$, we sample each edge adjacent to a vertex $v\in A$ with probability
$C\cdot\frac{k}{\sqrt{\deg_H(v)|A|}}$,
where $C$ is a constant to be chosen later.
As the size of the component containing a vertex~$v$ at least doubles between two consecutive growth steps in which the component of~$A$ participates, it is easy to check that the total sampling probability is not exceeded.

\subsection{Growth of high-degree components}\label{sub-high}

When $k\ge c\log n$, high-degree components can be easily grown using a simple technique.

To grow a high-degree component $A$ we pick a vertex $v\in A$ with $d=\deg_{H}(v) \geq 2|A|$ and perform rounds in which we sample each edge of~$v$ with probability~$1/d$. We do that until an edge connecting~$A$ to another component is sampled. The next lemma shows that with high probability no sampling probability is exceeded, assuming that $k\ge c\log n$.

\begin{lemma}\label{smpls1}
With high probability, the over all probability in which an edge touching a vertex of degree~$d$ is sampled during all growth steps of high-degree components is $O((\log n)/d)$.
\end{lemma}

\begin{proof}
Let $v$ be a vertex of degree~$d$. In each sampling round in the growth of a high-degree component in which~$v$ is the chosen highest degree vertex, we sample each edge adjacent to~$v$ with probability $1/d$. The probability that at least one edge is sampled in such a round is at least $1-(1-\frac{1}{d})^d>1-\frac{1}{\ee}>\frac{1}{2}$. If an edge is sampled, then with probability at least~$1/2$ it leaves the component and the process stops. Thus, each round is successful with a probability of at least~$1/4$.

\newcommand{\PP}{{\mathbb{P}}}

Each time~$v$ belongs to a component that undergoes a growth step, the size of the component of~$v$ doubles. Thus, $v$ participates in at most $\log n$ such steps. The total number of $1/d$ sampling rounds in which $v$ participates, over all growth steps, is thus stochastically dominated by a negative binomial random variable $T\sim NB(\log n,\frac{1}{4})$, the number of times a coin that comes up ``head'' with probability $1/4$ needs to be flipped until seeing $\log n$ ``heads''. We want to bound $\PP[T\ge 32\log n]$.

Let $X\sim B(32\log n,\frac{1}{4})$ be a binomial random variable that gives the number of ``heads'' in $32\log n$ throws of a coin that comes up ``head'' with probability $\frac{1}{4}$. Clearly $\PP[T\ge 32\log n]=\PP[X\le \log n]$.

By Chernoff's bound, if $X$ is a binomial random variable with $\mu=\EE[X]$ and $0<\delta<1$, then $\PP[X\le (1-\delta)\mu]\le \ee^{-\delta^2\mu/2}$. In our case $\mu=8\log n$ and $\delta=\frac{7}{8}$. Thus, $\PP[X\le \log n] \le \ee^{-(\frac{7}{8})^2\cdot 4\log n}\le \ee^{-3\log n}\le n^{-3}$. (Note that $\log n = \log_2 n > \ln n$.)

Thus by a simple union bound, the probability that any vertex, and hence any directed edge, participates in more than $32\log n$ rounds is at most $n^{-2}$.
\end{proof}

When $k\ge 32\log n$,  then $(32\log n)/d \le k/d$, so the probability of exceeding any sampling probability during the growth of high-degree components is at most $n^{-2}$. This is much more than what we need.

With a very low probability, our scheme might exceed the allowed sampling probability of certain edges. This is a slight deviation from the framework of Section~\ref{sub-overview}. This is justified as follows. Let~$S$ be the event that no sampling probability is exceeded. By Lemma~\ref{smpls1}, we have that $\mathbb{P}[S]\ge 1-n^{-2} \ge 1/2$. We prove below that $\EE[X]=O(n/k)$, where $X$ is the number inter-component edges. To prove Theorem~\ref{T-variant}, we need to show that $\EE[X|S]=O(n/k)$. However, as $\EE[X]\ge \mathbb{P}[S]\cdot\EE[X|S]$, we have $\EE[X|S]\le \EE[X]/\mathbb{P}[S]\le 2\,\EE[X]=O(n/k)$, as required.

\subsection{Main challenge: Growth of low-degree components}\label{sub-low}

To grow a low-degree component $A$ we sample each edge adjacent to a non-trimmed vertex $v\in A$ with probability
$C\cdot\frac{k}{\sqrt{\deg_H(v)|A|}}$,
where $C$ is a constant to be chosen later.
\begin{lemma}\label{smpls2}
The probability with which an edge $e$ adjacent on vertex~$v$ is sampled from~$v$ during growth steps of low-degree components is $O\left(\frac{k}{\deg_{G}(v)}\right)$.
\end{lemma}
\begin{proof}
Let $A_0\subset A_1\subset\cdots\subset A_{r-1}$ be all the low-degree components containing $v$ throughout the entire process until~$v$ is trimmed. Let~$H_i$ be the graph~$H$ when component~$A_i$ is grown. Note that $\deg_{H_i}(v)\ge \frac{1}{3}\deg_G(v)$, for every~$i$.
The probability with which~$e$ is sampled when growing~$A_i$ is $C\cdot\frac{k}{\sqrt{\deg_{H_i}(v)|A_i|}} = O( \frac{k}{\sqrt{\deg_{G}(v)|A_i|}}) $. We also have $|A_i|\ge 2^i |A_0|$.
As $A_0$ is a low-degree component we have $\deg_{H_0}(v)<2|A_0|$. Thus, the total probability with which~$e$ is sampled from $v$ is
\[ O\left(
\sum_{i=0}^{r-1} \frac{k}{\sqrt{\deg_{G}(v)|A_i|}}
\right) \;=\; \frac{k}{\deg_G(v)} O\left(\sum_{i=0}^{r-1} 2^{-i/2}\right) \;=\;
O\left( \frac{k}{\deg_G(v)} \right) \;. \]
\end{proof}

\newcommand\boundary{\partial}
For a component $A$, denote by $|\boundary_H A|$ the number of edges in
$H$ in the cut defined by $A$.

We need to show that $\mathbb{E}[|G\setminus H|] = O(n/k)$. As a warm-up,
we show that $\mathbb{E}[|G\setminus H|] = O((n\log n)/k)$.


%
%

\smallskip
\begin{proof}[That $\mathbb{E}[|G\setminus H|] = O((n\log n)/k)$.]
Let $A$ be the low-degree component that we are currently growing. We are sampling each edge, and in particular each edge of the cut, with probability
$C\cdot\frac{k}{\sqrt{\deg_H(v)|A|}}
>C\frac k{2|A|}$. If none of the edges of the cut $\partial_H(A)$ is sampled, then all the edges in the cut are added to~$H$. The \emph{expected} number of edges added to~$H$, divided by~$|A|$ is:
\[ \frac{|\partial_H(A)|}{|A|}\left(1-C\frac k{2|A|}\right)^{|\boundary_H(A)|} \;<\;
\frac{|\partial_H(A)|}{|A|} \exp\left(- \frac{Ck}{2}\frac{|\partial_H(A)|}{|A|}\right)\;.\]
The function $f(x)=x\ee^{-ax}$ is maximized at $x_0=\frac{1}{a}$ and its maximum value is $\frac{1}{\ee a}$. Thus, the maximum value of the expression above is $\frac{2}{\ee Ck}<\frac{1}{Ck}$.

Thus to cover this expected cost, each vertex only needs to pay $\frac{1}{Ck}$ each time it participates in a growing component. As this happens at most $\log n$ times, the total cost per vertex is at most $\frac{\log n}{Ck}$, and the total cost for all vertices is $\frac{n\log n}{Ck}$.
\end{proof}

\newcommand{\maxdeg}{d}

To prove that $\mathbb{E}[|G\setminus H|] = O(n/k)$, we need a much more
elaborate argument. We start with several definitions.

\begin{definition}[Maximum degree]
The \emph{maximum degree} $\maxdeg_H(A)$ of a component~$A$ is defined to be
\[ \maxdeg_H(A) \;:=\; \max_{v\in A}\deg_{H}(v) \;. \]
We also let $v(A)$ be a vertex of maximum degree in~$A$. (Ties broken arbitrarily.)
\end{definition}

\begin{definition}[Density]
The \emph{density} $\rho(A)$ of a component~$A$ is defined to be
\[ \rho(A)
\;:=\; |\partial_H A| / \Bigl(\frac{|A|}{k}\Bigr) = \frac{k|\partial_H A|}{|A|}\;. \]
\end{definition}

\begin{definition}[Density level of a component]
  The \emph{density level $\ell(A)$} of a component $A$ is defined to be the unique integer~$\ell$ such that $\rho(A) \in [2^\ell, 2^{\ell+1})$.
\end{definition}

\begin{definition}[Density level of a vertex]
  Let $v\in V$ and let $A_1\subset A_2\subset \cdots \subset A_r$ be all the low-degree components to which $v$ belonged so far. The \emph{level} $\ell(v)$ is defined to be
  \[ \ell(v) \;:=\; \min_{1\leq i \leq r} \ell(A_i)\;,\]
  where $\ell(A_i)$ is the level of~$A_i$ when it participated in a growth step. If $v$ did not participate yet in a low-degree component, then $\ell(v)=\infty$. Note that $\ell(v)$ 
  cannot increase. 
\end{definition}


\begin{definition}[Cost of a component]
Let $A$ be a low-degree component currently participating in a growth step.
We define the \emph{cost} of $A$ to be
\[
cost(A) \;:=\;
2^{\ell(A)} \frac{|A|}{k}\exp \left(
 -C 2^{\ell(A)} \left( \frac{|A|}{\maxdeg_H(A)} \right)^{1/2}
\right)
\;.
\]
\end{definition}

The following simple technical lemma is used several times in what follows.

\begin{lemma}\label{middleschoolmath}
The function $f(x)=x^\beta e^{-a x^{\gamma}}$, where $a,\beta,\gamma>0$, attains its maximum value at the point $x_0=(\frac{\beta}{a\gamma})^{\frac{1}{\gamma}}$
and is decreasing for $x\geq x_0$.
\end{lemma}
\begin{proof}
The claim is immediate as
\[
\frac{\partial f}{\partial x} = \beta x^{\beta -1} e^{-ax^{\gamma }} - x^\beta  \cdot a \gamma  x^{\gamma -1} e^{-ax^{\gamma }}
= x^{\beta -1} e^{-ax^{\gamma }} \left( \beta  - a\gamma x^{\gamma } \right) \;. \]
\end{proof}


\begin{lemma}\label{cutcost}
Let $A$ be a low-degree component. Then, the expected number of edges removed from~$H$ during the growth step of~$A$ is at most $\rho(A)\cdot \frac{|A|}{k}\cdot \exp \left(
 -C \rho(A) \left( \frac{|A|}{\maxdeg_H(A)} \right)^{1/2} \right)$.
\end{lemma}

\begin{proof} Let $d=\maxdeg_H(A)$ and let
 $q=C\cdot\frac{k}{\sqrt{d|A|}}$.
Every non-trimmed vertex $v\in A$ samples each of its incident edges with probability at least~$q$.
Thus, the probability of missing all cut edges is
\[
(1-q)^{\rho(A)\cdot \frac{|A|}{k}} \;\leq\; \exp\left(-q \rho(A)\cdot \frac{|A|}{k}\right) \;=\;
\exp \left(
 -C \rho(A) \left( \frac{|A|}{d} \right)^{1/2} \right) \;.
\]
The expected number of edges removed from~$H$ is thus at most
\[ \rho(A)\cdot \frac{|A|}{k}\cdot \exp \left(
 -C \rho(A) \left( \frac{|A|}{d} \right)^{1/2} \right) \;. \]
\end{proof}

We partition the growth steps of a low-degree component~$A$ into the following four types:

\begin{itemize}
  \item[] {\bf Type 0:} $\rho(A)\le 1$.
  \item[] {\bf Type 1:} $\maxdeg_H(A)\le 5|\partial_H A|$.
  \item[] {\bf Type 2:} $A$ is the first low-degree component of level at most~$\ell(A)$ that contains $v(A)$, a vertex of maximum degree in~$A$, and $\maxdeg_H(A)> 5|\partial_H A|$.
  \item[] {\bf Type 3:} $A$ is \emph{not} the first low-degree component of level at most~$\ell(A)$ that contains $v(A)$.
\end{itemize}

In types 1, 2 and 3 we assume that $\rho(A)>1$. In type~3 we may assume that $\maxdeg_H(A)> 5|\partial_H A|$, but we do not rely on it.

When growing a type 0 component~$A$, i.e., when $\rho(A)\le 1$, we do not actually try to grow~$A$. We simply remove all the edges of $\partial_H A$ from~$H$.

\begin{lemma}\label{T0-total} The total number of edges removed from~$H$ during growth steps of low-degree components of type~$0$ is at most $n/k$.
\end{lemma}

\begin{proof}
  Let $A_1,A_2,\ldots,A_r$ be all the low-degree type 0 components encountered throughout the process. As type 1 growth steps always fail, all these components are disjoint. Thus, the total number of edges removed from~$H$ is at most $\sum_{i=1}^r \frac{|A_i|}{k}\le \frac{n}{k}$.
\end{proof}

When growing a component~$A$ not of type~0, we sample each edge adjacent to a non-trimmed vertex of~$A$ with probability $C\frac{k}{\sqrt{\deg_H(v)|A|}}$. The next lemma justifies the definition of $cost(A)$.

\begin{lemma} 
Let $A$ be a low-degree component with $\rho(A)>1$. Then, the expected number of edges removed from~$H$ during the growth step of~$A$ is at most $cost(A)$.
\end{lemma}

\begin{proof} Let $d=\maxdeg_H(A)$. By Lemma~\ref{cutcost}, the expected number of edges removed while growing~$A$ is at most
\[
\rho(A)\cdot \frac{|A|}{k}\cdot \exp \left(
 -C \rho(A) \left( \frac{|A|}{d} \right)^{1/2} \right)
\;\le\;
 2^{\ell(A)}\cdot \frac{|A|}{k}\cdot \exp \left(
 -C 2^{\ell(A)} \left( \frac{|A|}{d} \right)^{1/2} \right)
 \;=\; cost(A)
 \;,\]
where the inequality follows from Lemma~\ref{middleschoolmath} as $\rho(A)\ge 2^{\ell(A)}\ge 1$. (Use the lemma with $x=\rho(A)$, $\beta=\gamma=1$ and $a=C(\frac{|A|}{d})^{1/2}$. Then, $x_0=\frac{1}{C}(\frac{d}{|A|})^{1/2}$. Note that $\frac{d}{|A|}\le 2$, so for $C\ge \sqrt{2}$ we have $x_0\le 1$, the function is decreasing for $x\ge 1$, and the inequality follows.)
\end{proof}

We next bound the expected cost of all growth steps of types 1,2 and 3.

\begin{lemma}\label{chargeall}
If $A$ is a low-degree component and $\maxdeg_H(A)\le 5 |\partial_H A|$, then
$\frac{cost(A)}{|A|}=o(\frac{1}{k \log n})$.
\end{lemma}

\begin{proof} Let $d=\maxdeg_H(A)$. Then,
\[
\frac{|A|}{d} \;\ge\; \frac{|A|}{5 |\partial_H A|} \;=\;
\frac{k}{5\rho(A)} \;>\;
\frac{k}{5\, 2^{\ell(A)+1}} \;.
\]
Therefore,
\begin{align*}
\frac{cost(A)}{|A|} &\;=\;
\frac{2^{\ell(A)}}{k}  \exp \left(
 -C 2^{\ell(A)} \left( \frac{|A|}{d} \right)^{1/2}
\right) \\
&\le\;
\frac{2^{\ell(A)}}{k}  \exp \left(
 -C 2^{\ell(A)} \left( \frac{k}{5\, 2^{\ell(A)+1}} \right)^{1/2} \right) \\
 &\le\;
\frac{2^{\ell(A)}}{k}  \exp \left(
 -C\left(\frac{k}{10}\right)^{1/2}\cdot 2^{\ell(A)/2}  \right)
 \;=\; \frac{f(x)}{k} \;,
\end{align*}
where $f(x)=x^\beta {\rm e}^{-ax^\gamma}$, with
$x = 2^{\ell(A)}$, $\beta = 1$, $\gamma = \frac{1}{2}$ and
$a = C\left(\frac{k}{10}\right)^{1/2}$. By Lemma~\ref{middleschoolmath}, $f(x)$ is decreasing for $x\ge x_0$ where $x_0=(\frac{\beta}{a\gamma})^{1/\gamma}={\left(\frac{2}{C}\right)}^2\frac{10}{k}<1$.
Since $x\ge 1$, $f(x)<f(1)$ and thus
\[
\frac{cost(A)}{|A|} \;=\; \frac{f(x)}{k} \;\le \; \frac{f(1)}{k} \;\le\; \frac{1}{k}\exp\left(-C\left(\frac{k}{10}\right)^{1/2}\right) \;=\; o\left(\frac{1}{k\log n}\right)\;,
\]
as $k=c\log n = \omega((\log\log n)^2)$.
\end{proof}

We note that the constant 5 in the statement of the lemma is arbitrary, but this is what we will use below.

\begin{lemma}\label{T1-total} The total cost of growth steps of type $1$ is
$o(n/k)$ 
.
\end{lemma}

\begin{proof}
  By Lemma~\ref{chargeall}, assigning a cost of 
  $\frac{cost(A)}{|A|}=o(\frac{1}{k \log n})$
  to each vertex of a component~$A$ participating in a growth step of type 1 covers the cost of this growth step. Since each vertex participates in at most $\log n$ such growth steps, the total cost is 
  $n\log n\cdot o(\frac{1}{k\log n}) = o(\frac{n}{k})$. 
\end{proof}

It remains to analyze growth steps of types 2 and 3.
We say that the level of a vertex $v\in A$ decreased to~$\ell(A)$, if $A$ is the first low-degree component containing~$v$ whose level is at most~$\ell(A)$. In particular, if~$A$ is of type~2, then the level of~$v(A)$ decreased to~$\ell(A)$.



\begin{lemma}\label{T2}
  If the growth step corresponding to~$A$ is of type $2$, then charging $\frac{2^{-\ell(A)}}{k}$ to every vertex of~$A$ whose level decreased to~$\ell(A)$ is enough to cover $cost(A)$.
\end{lemma}

\begin{proof} Let $v=v(A)$ be a vertex of degree $d=\maxdeg_H(A)$ in~$A$. As the growth step is of type 2, we have $\ell(v)=\ell(A)$ and $d>5|\partial_H A|$.

We start by showing that at least half of the neighbors of $v$ in $A$
also had their level decreased to~$\ell=\ell(A)$ by~$A$. Assume, for the sake of contradiction, that the level of at least half of neighbors of~$v$ in~$A$ was at most~$\ell$ before~$A$ was formed. Let $B_1,B_2,\ldots,B_s$ be the maximal components of level at most~$\ell$ that are included in~$A$. Clearly $B_1,B_2,\ldots,B_s$ are disjoint and they must contain at least half of the neighbors of~$v$ in~$A$. The number of neighbors of~$v$ in these components is at most $\sum_{i=1}^s |\partial_{H_i}B_i| \le 2^{\ell+1} \sum_{i=1}^s \frac{|B_i|}{k}\le 2^{\ell+1}\frac{|A|}{k}\le 2\rho(A)\frac{|A|}{k}=  2|\partial_H A|$, since the levels of the $B_i$'s is at most~$\ell$, and the level of~$A$ is~$\ell$. Thus, the total number of neighbors of~$v$ in~$A$ is at most $4|\partial_A H|$, and the total number neighbors of~$v$, not necessarily in~$A$, is at most $5|\partial_A H|$, a contradiction.

Thus, $v$ has at least $\frac{1}{2}(d-|\partial_H A|)\ge \frac{2}{5}d$ neighbors in~$A$ whose level decreased to~$\ell$ by~$A$. To cover $cost(A)$, it is thus enough to charge each one of these vertices by
\[ \frac{cost(A)}{\frac{2}{5}d} \;=\; \frac{2^\ell}{\frac{2}{5}k} \frac{|A|}{d} \exp\left(-C 2^\ell \left(\frac{|A|}{d}\right)^{1/2}\right) \;=\;
\frac{2^\ell}{\frac{2}{5}k}\cdot f\left(\frac{|A|}{d}\right)\;,
\]
where $f(x)=x\ee^{-ax^{1/2}}$ and $a=C2^\ell$. By Lemma~\ref{middleschoolmath}, with $\beta=1$ and $\gamma=\frac{1}{2}$, $f(x)$ attains its maximum at $x_0=(\frac{\beta}{a\gamma})^{1/\gamma}=(\frac{2}{C2^\ell})^2$. If $C\ge 3$, then $x_0< \frac{1}{2}< |A|/d$, for every $\ell\ge 0$. Thus,
\[
\frac{cost(A)}{\frac{2}{5}d} \;\le\;
\frac{2^\ell}{\frac{2}{5}k} f(x_0) \;=\;
\frac{2^\ell}{\frac{2}{5}k}\Bigl(\frac{2}{C 2^\ell}\Bigr)^2 \ee^{-2} \;=\;
\frac{10}{(C\ee)^2}\cdot\frac{2^{-\ell}}{k} \;<\;
\frac{2^{-\ell}}{k}
\]

\end{proof}

\begin{lemma}\label{T2-total} The total cost of steps of type~$2$ is $O(n/k)$.
\end{lemma}

\begin{proof} By Lemma~\ref{T2}, charging every vertex whose level decreases to~$\ell$ a cost of $\frac{2^{-\ell}}{k}$ covers the cost of all growth steps of type~2. As the level of each vertex decreases to any given level at most once, we get that the total charge for a given vertex is at most $\sum_{\ell\ge 0} \frac{2^{-\ell}}{k}\le \frac{2}{k}$. The total charges of these form, for all vertices, is thus at most $\frac{2n}{k}$.
\end{proof}


\begin{lemma}\label{AB} Suppose that the growth step of~$A$ is of type $3$, i.e., the vertex~$v=v(A)$ of maximum degree in~$A$ is contained in a previous component~$B$ with $\ell(B)\le \ell(A)$. Then, $cost(A) \leq \frac{|B|}{|A|} cost(B)$.
\end{lemma}

\begin{proof} Let $d=\maxdeg_H(A)$ and $d'=\maxdeg_{H'}(B)$, where~$H'$ is the graph~$H$ at the time $B$ was grown.
Note that $d\le 2|A|$ and $d'\le 2|B|$, as both $A$ and $B$ are low degree components. As $v$ is a vertex of maximum degree in~$A$ and as $v\in B$, we get that $d'\ge d$. (Note that the degree of~$v$ may have decreased, which works in our favor.) Let $\ell'=\ell(B)\le \ell(A)=\ell$. Then,
\[\textstyle cost(B) \;=\; 2^{\ell'} \frac{|B|}{k}\exp \!\left(
 -C 2^{\ell'} \left( \frac{|B|}{d'} \right)^{\!\frac{1}{2}}
\right)
\;\ge\; 2^{\ell'} \frac{|B|}{k}\exp \!\left(
 -C 2^{\ell'} \left( \frac{|B|}{d} \right)^{\!\frac{1}{2}}
\right)
\;\ge\; 2^{\ell} \frac{|B|}{k}\exp \!\left(
 -C 2^{\ell} \left( \frac{|B|}{d} \right)^{\!\frac{1}{2}}
\right) \;,
\]
where the last inequality follows from Lemma~\ref{middleschoolmath}. (Let  $x=2^\ell$, $\beta=\gamma=1$ and $a=C\left( \frac{|B|}{d} \right)^{1/2}$. The maximum is attained at $x_0=(\frac{\beta}{\gamma}a)^{1/\gamma} = \frac{1}{C}\left( \frac{d}{|B|} \right)^{1/2}$. As $\frac{d}{|B|}\le 2$, we get that $x_0\le 1$ for $C\ge\sqrt{2}$.)

We next claim that
\[ \textstyle
 |A|cost(A) \;=\; 2^{\ell} \frac{|A|^2}{k}\exp \left(
 -C 2^{\ell} \left( \frac{|A|}{d} \right)^{1/2}
\right) \;\le\;
2^{\ell} \frac{|B|^2}{k}\exp \left(
 -C 2^{\ell} \left( \frac{|B|}{d} \right)^{1/2}
\right) \;\le\;
|B|cost(B)\;.
\]
The first inequality follows again using Lemma~\ref{middleschoolmath}. (Divide both sides by $d^2$. Let $x=\frac{|A|}{d}$, $\beta=2$, $\gamma=\frac{1}{2}$ and $a=C2^\ell$. Then, $x_0=(\frac{\beta}{a\gamma})^{1/\gamma} = (\frac{4}{C2^\ell})^2$. Thus, $x_0\le\frac{1}{2}$, for every $\ell\ge 0$, when, say, $C\ge 6$. Note that $\frac{1}{2}\le\frac{|B|}{d}\le\frac{|A|}{d}$.) The claim of the lemma follows.
\end{proof}

\begin{lemma}\label{T3} Suppose that the growth step of~$A$ is of type~$3$. Then there is a component $B\subset A$ whose growth step is not of type~$3$ such that
$cost(A) \leq \frac{|B|}{|A|} cost(B)$.
\end{lemma}

\begin{proof} By induction on the order which the components were created. The first low-degree component is not of type~$3$, which forms the basis of the induction. Suppose that $A$ is a type~3 component and that the claim holds for every component created before~$A$. By Lemma~\ref{AB}, there exists a component $B\subset A$ such that $cost(A) \leq \frac{|B|}{|A|} cost(B)$. If~$B$ is not of type~$3$, we are done. Otherwise, by the induction hypothesis, there exists a component~$C\subset B\subset A$, where $C$ is not of type~$3$, such that $cost(B) \leq \frac{|C|}{|B|} cost(C)$. We then have
\[cost(A) \;\leq\; \frac{|B|}{|A|} cost(B) \;\le\; \frac{|B|}{|A|}\frac{|C|}{|B|}cost(C) \;=\; \frac{|C|}{|A|}cost(C)\;,\]
as required.
\end{proof}

\begin{lemma}\label{T3-total}
  The total cost of all growth steps of type~$3$ is at most the total cost of all growth steps \emph{not} of type~$3$.
\end{lemma}

\begin{proof} By Lemma~\ref{T3}, if $A$ is of type~3, then there exists a component $B\subset A$, not of type~3, such that $cost(A) \le \frac{|B|}{|A|} cost(B)$. We thus charge the cost of~$A$ to~$B$. Let $B$ be non-type 3 component, and let $B\subset A_1\subset A_2\subset\cdots\subset A_s$ be the type 3 components whose cost is charged to~$B$. As $|A_i|\ge 2^i|B|$, we get that the total charge for~$B$ is at most $(\sum_{i\ge 1} 2^{-i})cost(B)\le cost(B)$.
\end{proof}

Combining Lemmas~\ref{T0-total}, \ref{T1-total}, \ref{T2-total} and~\ref{T3-total} we obtain the following theorem which implies Theorem~\ref{T-main}.

\begin{theorem}
The total cost of all low-degree growth steps is $O(n/k)$.
\end{theorem}

As an immediate corollary of Theorem~\ref{T-main} we obtain:

\begin{corollary}\label{C-tail}
  Let $G=(V,E)$ be an arbitrary undirected $n$-vertex graph and let $k\ge c\log n$, where~$c$ is a large enough constant. Let $G'$ be a random $k$-out subgraph of~$G$. 
  Then there exists a constant $b$ such that the probability that the number of edges in~$G$ that connect different connected components of~$G'$ exceeds $\ell\cdot bn/k$ is at most $2^{-\ell}$.
\end{corollary}

\begin{proof}
  By Theorem~\ref{T-main} there is a constant $b$ such that the expected number of inter-component edges is at most $bn/k$. By Markov's inequality, the probability that the number of inter-component edges is more than $2bn/k$ is at most~$1/2$.


  A random $k$-out subgraph can be obtained by taking the union of $\ell$ independent random $k/\ell$-out subgraphs, and adding more edges, if needed, to make sure that $k$ edges incident on each vertex were chosen.
  The number of inter-component edges with respect to the random $k$-out subgraph is clearly at most the number of such edges for each one of the $k/\ell$-out subgraph. Thus, the probability that the number of inter-component edges is more than $\ell\cdot 2bn/k$ is at most $2^{-\ell}$.
\end{proof}

It is not difficult to see that the tail bound given in Corollary~\ref{C-tail} is asymptotically tight.

\section{One-way spanning forest protocol with private randomness}\label{S-sfprivate}

As an application of the new sampling theorem we consider the following one-way communication problem. Each vertex of an input graph has a distinct ID of $O(\log n)$ bits. Each vertex knows its ID and the IDs of its neighbors. Each vertex can send a \emph{single} message to a \emph{referee}. The referee must then determine a spanning forest of the graph. How many bits does each vertex need to send?

A sketching technique of Ahn, Guha and McGregor \cite{ahn2012analyzing,ahn2012graph} (see also Gibb et al.~\cite{gibb2015dynamic}) provides a $O(\log^3 n)$-bit solution, provided that \emph{public} randomness is available. The referee can then determine a spanning forest with a constant probability. Nelson and Yu~\cite{NY2018} have recently shown that this bound is tight.

We provide the first $o(n)$-bit solution using \emph{private} randomness. More specifically, we show that each vertex only needs to send $O(\sqrt{n}\log n)$ bits, following which the referee can determine a spanning forest with constant probability. The failure probability can be made polynomially small if each vertex sends $O(\sqrt{n}\log^{3/2} n)$ bits.

In addition to the sampling theorem, we need two additional ingredients which are described next.

\subsection{XOR trick}\label{sub-xor}

The XOR trick is the basis of the sketching technique of Ahn et al.~\cite{ahn2012analyzing,ahn2012graph}. It is also used by Kapron et al.~\cite{KKM13} and Gibb et al.~\cite{gibb2015dynamic} to obtain dynamic graph connectivity algorithms. In the following we identify the name of a vertex with its ID.

Suppose that each edge $\{i,j\}$ is assigned an $\ell$-bit string \emph{name} $x(\{i,j\})$. (Edges are undirected, so $x(\eij)=x(\{j,i\})$.) Perhaps the most natural name of an edge is the concatenation of the names of its endpoints, in an appropriate order. We will usually employ, however, more `resilient' edge names, as explained in Section~\ref{sub-resilient}.

For each vertex $i\in V$, we let $X(i)=\bigoplus_{j:\{i,j\}\in E} \xij$. For any $C\subset V$, let $X(C)=\bigoplus_{i\in C} X(i)$. Recall that $\partial C$ denotes the set of edges that cross the cut $(C,V\setminus C)$.

\begin{lemma} For any $C\subset V$, we have
$ X(C) = \bigoplus_{\{i,j\}\in \partial C} \xij$. In words, $X(C)$ is the xor of the names of all the edges that cross the cut $(C,V\setminus C)$.
\end{lemma}

In particular, if $|\partial C|=1$, then $X(C)$ is the name of the single edge in the cut $(C,V\setminus C)$. This simple observation is heavily used in \cite{ahn2012analyzing,ahn2012graph,gibb2015dynamic,KKM13}.

\subsection{Resilient edge names}\label{sub-resilient}

We observed above that if $|\partial C|=1$, then $X(C)$ is the name of the single edge in the cut $(C,V\setminus C)$. What if $|\partial C|>1$? To identify the edges that cross the cut in this case, provided that there are at most $k$ edges that cross the cut we use \emph{resilient} edge names.

\begin{definition}[Resilient edge names]
A collection of edge names is said to be $r$-\emph{resilient}, if and only if, for any two subsets $A\ne B\subset E$, with $|A|,|B|\le r$ we have $X(A)\ne X(B)$, where $X(A)=\bigoplus_{\{i,j\}\in A}\xij$. Equivalently, for every $A\subset E$ with $|A|\le 2r$, we have $X(A)\ne 0$.
\end{definition}

\begin{lemma}\label{L-resilient} For every $r\ge 1$, the edges of a complete undirected graph on $n$ vertices can be given $r$-resilient edge names of length $\ell \le 4r \lg n$.
\end{lemma}

\begin{proof} We use a simple probabilistic argument. Let $\ell = 4r \lg n$. For every $i<j$, let $\xij$ be a random $\ell$-bit string. All bit strings are chosen independently. For a given $A\subset E$, the probability that $X(A)=0$ is exactly $2^{-\ell}$. By the union bounds, the probability that there exists a set of at most $2r$ edges $A$ such that $X(A)=0$ is at most $2^{-\ell}\sum_{i=1}^{2r} { {n \choose 2} \choose i } < 1$.
\end{proof}

In Appendix~\ref{A-codes} we describe efficient explicit construction of $r$-resilient edge names based on linear error correcting codes. We note that for our purposes the \emph{existence} of $r$-resilient names is enough, as the vertices can agree on such names before the protocol starts.

\subsection{$O(\sqrt{n}\log n)$-bit messages using private randomness}\label{sub-private}

The $n$ vertices agree on a collection of $r$-resilient names for all potential edges, where $r=c\sqrt{n}$, for a sufficiently large constant $c$. Each edge thus has an $\ell$-bit name $\xij$, where $\ell=O(\sqrt{n}\log n)$. 

The message vertex~$i$ sends to the referee is composed of two parts:
\begin{enumerate}
\item A sample of $\sqrt{n}$ edges incident on~$i$, or all the edges incident on~$i$, if its degree is less than~$\sqrt{n}$. (Note that this corresponds exactly to the $\sqrt{n}$-out model.) This part is composed of $O(\sqrt{n}\log n)$ bits.
\item The xor of the names of all the edges incident on the vertex, i.e., $X(i)=\bigoplus_{j:\{i,j\}\in E} \xij$. The number of bits in this part is again $O(\sqrt{n}\log n)$.
\end{enumerate}

The collection of edges received by the referee is exactly a random $\sqrt{n}$-out subgraph~$G'$ of the original graph. The referee computes the connected components and a spanning forest of this subgraph. By Theorem~\ref{T-main}, with $k=\sqrt{n}$, the expected number of inter-component edges is $O(\sqrt{n})$. Thus, with probability at least $1/2$ the number of inter-component edges is at most $c\sqrt{n}$, for some constant~$c$. We assume in the following that this is the case.

For every connected component~$C$ of~$G'$, the referee computes $X(C)=\bigoplus_{i\in C} X(i)$. As the number of inter-component edges is at most~$c\sqrt{n}$, we also have $|\partial C|\le c\sqrt{n}$. As the names of the edges are $c\sqrt{n}$-resilient, the referee can infer all the edges of~$\partial C$. The referee can thus easily extend the spanning forest of~$G'$ to a spanning forest of~$G$.

The protocol described produces a spanning forest of the input graph with probability of at least~$1/2$. When it fails, the spanning forest returned
may contain edges not present in the graph, and may fail to span all connected components of the graph. By Corollary~\ref{C-tail}, the failure probability of the algorithm can be reduced to $2^{-\ell}$, for any $\ell\ge 1$, by using $\ell\cdot c\sqrt{n}$ resilient edge names, or alternatively, sending $\ell\cdot c\sqrt{n}$ edges incident on each vertex.

To get a polynomially small error probability, i.e., $n^{-\alpha}$, for some $\alpha>0$, each vertex sends $a\sqrt{n\log n}$ incident edges, and $a \sqrt{n\log n}$-resilient edge names are used, for a sufficiently large constant~$a$.

The spanning forest protocol gives a Monte-Carlo, i.e., two-sided error, protocol for checking the connectivity of a graph. Converting the protocol into a Las Vegas, i.e., one-sided error, protocol is an interesting open problem.

We note that an $\tilde{O}(\sqrt{n})$ protocol with private randomness can also be obtained without the use of the new sampling theorem. However, the procedure used by the referee to construct a spanning forest of the input graph is more complicated, and the number of bits sent by each vertex is larger by a factor of $\log n$. See Appendix~\ref{A-protocol} for the details.


\section*{Acknowledgment} The last author would like to thank Orr Fischer and Rotem Oshman for introducing him to the communication complexity problem and for many helpful discussions about it.


\begin{thebibliography}{10}

\bibitem{ahn2012analyzing}
Kook~Jin Ahn, Sudipto Guha, and Andrew McGregor.
\newblock Analyzing graph structure via linear measurements.
\newblock In {\em Proceedings of the twenty-third annual ACM-SIAM symposium on
  Discrete Algorithms}, pages 459--467. SIAM, 2012.

\bibitem{ahn2012graph}
Kook~Jin Ahn, Sudipto Guha, and Andrew McGregor.
\newblock Graph sketches: sparsification, spanners, and subgraphs.
\newblock In {\em Proceedings of the 31st ACM SIGMOD-SIGACT-SIGAI symposium on
  Principles of Database Systems}, pages 5--14. ACM, 2012.

\bibitem{bohman2009hamilton}
Tom Bohman and Alan Frieze.
\newblock Hamilton cycles in 3-out.
\newblock {\em Random Structures \& Algorithms}, 35(4):393--417, 2009.

\bibitem{Chan98}
Timothy~M. Chan.
\newblock Backwards analysis of the {Karger-Klein-Tarjan} algorithm for minimum
  spanning trees.
\newblock {\em Inf. Process. Lett.}, 67(6):303--304, 1998.

\bibitem{FennerF82}
Trevor~I. Fenner and Alan~M. Frieze.
\newblock On the connectivity of random $m$-orientable graphs and digraphs.
\newblock {\em Combinatorica}, 2(4):347--359, 1982.

\bibitem{frieze2016introduction}
Alan Frieze and Micha{\l} Karo{\'n}ski.
\newblock {\em Introduction to random graphs}.
\newblock Cambridge University Press, 2016.

\bibitem{frieze1986maximum}
Alan~M Frieze.
\newblock Maximum matchings in a class of random graphs.
\newblock {\em Journal of Combinatorial Theory, Series B}, 40(2):196--212,
  1986.

\bibitem{FriezeGRV14}
Alan~M. Frieze, Navin Goyal, Luis Rademacher, and Santosh Vempala.
\newblock Expanders via random spanning trees.
\newblock {\em {SIAM} J. Comput.}, 43(2):497--513, 2014.

\bibitem{FriezeJ17}
Alan~M. Frieze and Tony Johansson.
\newblock On random \emph{k}-out subgraphs of large graphs.
\newblock {\em Random Struct. Algorithms}, 50(2):143--157, 2017.

\bibitem{GNT19}
Mohsen Ghaffari, Krzysztof Nowicki, and Mikkel Thorup.
\newblock Faster algorithms for edge connectivity via random 2-out
  contractions.
\newblock {\em CoRR}, abs/1909.00844, 2019. Accepted for {\em SODA'20}.

\bibitem{gibb2015dynamic}
David Gibb, Bruce Kapron, Valerie King, and Nolan Thorn.
\newblock Dynamic graph connectivity with improved worst case update time and
  sublinear space.
\newblock {\em arXiv preprint arXiv:1509.06464}, 2015.

\bibitem{HalperinZ01}
Shay Halperin and Uri Zwick.
\newblock Optimal randomized {EREW} {PRAM} algorithms for finding spanning
  forests.
\newblock {\em J. Algorithms}, 39(1):1--46, 2001.

\bibitem{Jurdzinski2018}
Tomasz {Jurdzi\'{n}ski} and Krzysztof Nowicki.
\newblock {MST} in {O(1)} rounds of congested clique.
\newblock In {\em Proceedings of the Twenty-Ninth Annual ACM-SIAM Symposium on
  Discrete Algorithms}, SODA '18, pages 2620--2632. SIAM, 2018.

\bibitem{KKM13}
Bruce~M Kapron, Valerie King, and Ben Mountjoy.
\newblock Dynamic graph connectivity in polylogarithmic worst case time.
\newblock In {\em Proceedings of the twenty-fourth annual ACM-SIAM symposium on
  Discrete algorithms}, pages 1131--1142. Society for Industrial and Applied
  Mathematics, 2013.

\bibitem{KargerKT95}
David~R. Karger, Philip~N. Klein, and Robert~Endre Tarjan.
\newblock A randomized linear-time algorithm to find minimum spanning trees.
\newblock {\em J. {ACM}}, 42(2):321--328, 1995.

\bibitem{KargerNP99}
David~R. Karger, Noam Nisan, and Michal Parnas.
\newblock Fast connected components algorithms for the {EREW} {PRAM}.
\newblock {\em {SIAM} J. Comput.}, 28(3):1021--1034, 1999.

\bibitem{Karloff2010}
Howard Karloff, Siddharth Suri, and Sergei Vassilvitskii.
\newblock A model of computation for {MapReduce}.
\newblock In {\em Proceedings of the Twenty-first Annual ACM-SIAM Symposium on
  Discrete Algorithms}, SODA '10, pages 938--948. SIAM, 2010.

\bibitem{mauldin1981scottish}
R.~Daniel Mauldin.
\newblock {\em The Scottish book: mathematics from the Scottish Caf{\'e}}.
\newblock Birkhauser, 1981.

\bibitem{NY2018}
Jelani Nelson and Huacheng Yu.
\newblock Optimal lower bounds for distributed and streaming spanning forest
  computation.
\newblock {\em CoRR}, abs/1807.05135, 2018.

\end{thebibliography}

\appendix

\section{Resilient names and error correcting codes}\label{A-codes}

A binary $(n,k,d)$-code $\C$ is a $k$-dimensional linear subspace of $\mathbb{Z}_2^n$  such that for every $x,y\in \C$ we have $d_H(x,y)\ge d$, there $d_H(x,y)$ is the Hamming distance between the codewords $x$ and $y$. Each such code has an $(n-k)\times n$ \emph{parity check matrix} $A$ such that $x\in\C$ if and only if $Ax=0$. (The rows of~$A$ form a basis of the orthogonal subspace $\C^\perp$.)

If $A$ is the parity check matrix of an $(n,k,d)$-code $\C$, then as $0\in\C$, for every $x\in \mathbb{Z}_2^n$ with $w(x)<d$ we have $x\not\in \C$ and thus $Ax\ne 0$. (Here $w(x)=d_H(x,0)$ is the \emph{weight} of~$x$, i.e., the number of non-zero coordinates in~$x$.) In other words, the xor of any subset of less than~$d$ columns of~$A$ is non-zero. Thus, the columns of~$A$ form a collection of $\frac{d-1}{2}$-resilient names. To obtain a collection of $r$-resilient $\ell$-bit names for the edges of the complete graph we can use an $({n\choose 2},{n\choose 2}-\ell,2r+1)$-code.

For every $r$, the BCH code of length~$n$ is an $(n,n-r\log n,2r)$ code. It can thus be used to obtain an explicit $r$-resilient naming scheme with names of length $O(r\log n)$.


\section{An alternative protocol}\label{A-protocol}

In this Section we describe an alternative one-way communication  protocol with private randomness for the spanning forest problem that does \emph{not} rely on Theorem~\ref{T-main}. It is, however, slightly less efficient, and slightly more complicated. Trying to improve and simplify this algorithm led us to the discovery of Theorem~\ref{T-main}.

The protocol is similar to the protocol given in Section~\ref{sub-app}. Each vertex sends, however, $\log n$ independent samples of its edges, each edge is included in each one of the samples with probability $c/\sqrt{n}$. Each vertex also sends the xor of $c\sqrt{n}\log n$-resilient names of its incident edges,  for some large constant $c$. With high probability, each vertex sends at most $O(\sqrt{n}\log^2 n)$ bits. (Note that this is larger by a factor of $\log n$ compared to Section~\ref{sub-app}.)

The referee proceeds in $\log n$ rounds. In each round she only uses edges of the $i$-th sample which contains, in expectation, at most $c\sqrt{n}$ edges. (Note that unlike the protocol of Section~\ref{sub-app}, low degree vertices actually send shorter messages.)

In the beginning of each round the referee has a collection of components. At the start of the first round each vertex is its own component. Let $C$ be a component at the start of the $i$-th round. If there is an edge~$e$ of the $i$-th sample that connects~$C$ to a different component~$C'$, then the components~$C$ and~$C'$ are merged and~$e$ is added to the spanning forest. If no such edge is found, then we can `infer', with high probability, that the number of edges in the cut $(C,V\setminus C)$ is at most~$\sqrt{n}\log n$. (See justification in the next paragraph.) Indeed, if the size of the cut is at least~$\sqrt{n}\log n$, then the probability that none of the edges of the cut appears in the $i$-th sample is at most $(1-\frac{c}{\sqrt{n}})^{\sqrt{n}\log n} \le n^{-c}$. In this case, the referee computes $ X(C) = \bigoplus_{\{i,j\}\in \partial C} \xij$ from which she can infer the up to~$\sqrt{n}\log n$ edges of the cut. If $X(C)=0$, the referee declares $C$ to be a connected component of the input graph. Otherwise, it uses the edges returned to connect~$C$ to other components.

Let $C_1,C_2,\ldots,C_r$ be the components in the start of the $i$-th round. Suppose that $r'$ of the cut sets $\partial C_i$ are of size at least~$\sqrt{n}\log n$. The probability that the $i$-th sample fails to hit each one of these~$r_i'$ cuts is at most $r_i'\cdot n^{-c}$. Note that this is also the probability that the referee makes a mistake in the $i$-th round. If no mistake is made, then the number of components that are not complete components of the input graph decreases by a factor of at least~2 in each round, and after $\log n$ round, no such component remains. The referee outputs the spanning tree obtained. The total error probability is $(\sum r'_i) n^{-c}<n^{-(c-1)}$. Thus, choosing $c=2$ suffices to get a correct result with probability at least $1-\frac{1}{n}$.

It is important to note that the analysis above is correct as the cuts that we are trying to hit with the $i$-th sample are \emph{independent} of the $i$-th sample. (They only depend on the first $i-1$ samples.)

In the above protocol each edge is sampled independently with probability
$c/\sqrt{n}$. The \emph{same} sampling probability is used for all vertices. We have a simple example that shows that the $\log n$ samples used by the protocol are required in this case. Theorem~\ref{T-main} shows, perhaps surprisingly,
that if we sample $c\sqrt{n}$ edges from each vertex, i.e., giving low degree vertices a higher sampling probability, then the separate $\log n$ samples can be replaced by a single sample. Proving that, however, is far from easy.

\end{document}